\documentclass[a4paper,USenglish,cleveref,hyperref,autoref]{lipics-v2021}

\title{Coloring Hardness on Low Twin-Width Graphs}
\titlerunning{Coloring Hardness on Low Twin-Width Graphs}

\author{\'{E}douard Bonnet}{Univ Lyon, CNRS, ENS de Lyon, Université Claude Bernard Lyon 1, LIP UMR5668, France \and \url{http://perso.ens-lyon.fr/edouard.bonnet/}}{edouard.bonnet@ens-lyon.fr}{https://orcid.org/0000-0002-1653-5822}{}

\authorrunning{\'E. Bonnet}

\Copyright{Édouard Bonnet}

\category{}

\relatedversion{}

\supplement{}

\funding{This work has been supported by the French National Research Agency through the project TWIN-WIDTH with reference number ANR-21-CE48-0014.}

\acknowledgements{}

\nolinenumbers 

\hideLIPIcs  

\EventEditors{John Q. Open and Joan R. Access}
\EventNoEds{2}
\EventLongTitle{42nd Conference on Very Important Topics (CVIT 2016)}
\EventShortTitle{CVIT 2016}
\EventAcronym{CVIT}
\EventYear{2016}
\EventDate{December 24--27, 2016}
\EventLocation{Little Whinging, United Kingdom}
\EventLogo{}
\SeriesVolume{42}
\ArticleNo{23}

\usepackage[utf8]{inputenc}  

\usepackage[T1]{fontenc}
\usepackage{lmodern}

\usepackage[colorinlistoftodos,bordercolor=orange,backgroundcolor=orange!20,linecolor=orange,textsize=normalsize]{todonotes}

\usepackage{amsmath}  
\usepackage{amssymb}     
\usepackage{bm}
\usepackage{accents}
\usepackage{complexity}

\usepackage{booktabs}
\usepackage{paralist}
\usepackage{fixmath}

\crefname{question}{Question}{Questions}

\makeatletter
\newtheorem*{rep@theorem}{\rep@title}
\newcommand{\newreptheorem}[2]{%
\newenvironment{rep#1}[1]{%
 \def\rep@title{#2 \ref{##1}}%
 \begin{rep@theorem}}%
 {\end{rep@theorem}}}
\makeatother

\newreptheorem{theorem}{Theorem}
\newreptheorem{lemma}{Lemma}

\usepackage{xspace}
\usepackage{tikz}

\usepackage[ruled,vlined,linesnumbered]{algorithm2e}

\usetikzlibrary{fit}
\pgfdeclarelayer{bg}
\pgfsetlayers{bg,main}
\usetikzlibrary{arrows,arrows.meta}
\usetikzlibrary{patterns}
\usetikzlibrary{calc}
\usetikzlibrary{shapes}
\usetikzlibrary{positioning}
\usetikzlibrary{math}
\usetikzlibrary{shapes.symbols}
\usetikzlibrary{decorations.pathreplacing,calligraphy}
\usepackage[scr=boondox,scrscaled=1.05]{mathalfa}

\newcommand{\mis}{\textsc{Max Independent Set}\xspace}

\newcommand{\vc}{\textsc{Vertex Cover}\xspace}
\newcommand{\cli}{\textsc{Max Clique}\xspace}
\newcommand{\fvs}{\textsc{Feedback Vertex Set}\xspace}
\newcommand{\tcol}{\textsc{3-Coloring}\xspace}
\newcommand{\kcol}{\textsc{$k$-Coloring}\xspace}
\newcommand{\hamp}{\textsc{Hamiltonian Path}\xspace}
\newcommand{\hamc}{\textsc{Hamiltonian Cycle}\xspace}
\newcommand{\mds}{\textsc{Dominating Set}\xspace}
\newcommand{\mim}{\textsc{Max Induced Matching}\xspace}

\newcommand{\chr}{\textsc{Min Coloring}\xspace}

\newcommand{\bandw}{\textsc{Bandwidth}\xspace}

\renewcommand{\leq}{\leqslant}

\newtheorem{question}{Question}

\begin{document}

\maketitle

\begin{abstract} 
  As the class $\mathcal T_4$ of graphs of twin-width at~most~4 contains every finite subgraph of the infinite grid and every graph obtained by subdividing each edge of an $n$-vertex graph at~least~$2 \log n$ times, most \NP-hard graph problems, like \mis, \mds, \hamc, remain so on $\mathcal T_4$.
  However, \chr and \kcol are easy on both families because they are 2-colorable and 3-colorable, respectively.

  We show that \chr is \NP-hard on the class $\mathcal T_3$ of graphs of twin-width at~most~3.
  This is the first hardness result on~$\mathcal T_3$ for a~problem that is easy on cographs (twin-width~0), on trees~(whose twin-width is at~most~2), and on unit circular-arc graphs (whose twin-width is at~most~3).
  We also show that for every $k \geqslant 3$, \kcol is \NP-hard on~$\mathcal T_4$.
  We finally make two observations: (1) there are currently very few problems known to be in \P\ on~$\mathcal T_d$ (graphs of twin-width at~most~$d$) and \NP-hard on~$\mathcal T_{d+1}$ for some nonnegative integer~$d$, and (2) unlike $\mathcal T_4$, which contains every graph as an induced minor, the class $\mathcal T_3$ excludes a~fixed \emph{planar} graph as an induced minor; thus it may be viewed as a~special case (or potential counterexample) for conjectures about classes excluding a~(planar) induced minor.
  These observations are accompanied by several open questions.
\end{abstract}

\section{Introduction}\label{sec:intro}

The graph parameter \emph{twin-width} was introduced in~2020~\cite{twin-width1}.
The family of graph classes of bounded twin-width is broad and diverse: it includes classes of bounded clique-width, $d$-dimensional grids, classes excluding a~fixed minor~\cite{twin-width1}, some cubic expanders~\cite{twin-width2}.
Moreover, there are algorithms that work on any class of effectively\footnote{A~class $\mathcal C$ has \emph{effectively bounded twin-width} if there is an integer~$d$ and a~polynomial-time algorithm that outputs a~$d$-sequence (see~\cref{sec:prelim} for all relevant definitions) of any input graph from~$\mathcal C$.} bounded twin-width: a~fixed-parameter tractable first-order model checking algorithm~\cite{twin-width1}, single-exponential parameterized algorithms~\cite{twin-width3}, improved approximation algorithms~\cite{twin-width3,BergeBDW23}, fast shortest-path algorithms~\cite{twin-width3,Bannach24}, etc.
Most \NP-hard graph problems remain intractable on graphs of twin-width at~most~4.
This paper tackles the hardness of coloring graphs of low twin-width.

For any nonnegative integer~$d$, we denote by $\mathcal T_d$ the class of graphs of twin-width at~most~$d$.
Finite subgraphs of the infinite planar grid~\cite{twin-width1} and $(\geqslant 2 \log n)$-subdivisions of $n$-vertex graphs~\cite{BergeBonnetDepres} are in~$\mathcal T_4$.
Problems such as \mis, \vc, \cli, \fvs, \mds, and \mim are \NP-hard on the latter family, while \hamp and \hamc are \NP-hard on the former.
Therefore all these problems are \NP-hard on~$\mathcal T_4$.
Notably, \chr is easy on subgraphs of the grid, since they are bipartite.
It is also easy on \emph{strict} subdivisions (i.e., when every edge is subdivided at~least~once), since such graphs are always 3-colorable and bipartiteness can be checked in linear time.
We show that \chr is already \NP-hard on~$\mathcal T_3$.
This is the first problem that is shown \NP-hard on~$\mathcal T_3$ while being in~\P\ on classes known to have twin-width at~most~3 (cographs, trees, unit circular-arc graphs).

\begin{theorem}\label{thm:col-tww-3}
  \chr is \NP-hard and, unless the ETH fails, requires $2^{\Omega(\sqrt{n})}$ time on $n$-vertex graphs of twin-width at~most~3 (even if a~3-sequence is provided). 
\end{theorem}

In the previous theorem, the ETH (for Exponential-Time Hypothesis)~\cite{Impagliazzo01} asserts that there exists a~constant $\lambda>1$ such that $n$-variable \textsc{3-Sat} cannot be solved in time $O(\lambda^n)$.
The reduction requires the number of colors to grow with~$n$.

We also show that \tcol is \NP-hard on~$\mathcal T_4$.
Note that planar graphs have twin-width at~least~7~\cite{Kral25} (and at~most~8~\cite{Hlineny25}).
So we cannot just invoke the hardness of \tcol in planar graphs.
We tune long subdivisions to turn them into hard \tcol instances.

\begin{theorem}\label{thm:tcol-hardness}
  \tcol is \NP-hard on graphs of twin-width at~most~4.  
\end{theorem}

This again holds even if a~\mbox{4-sequence} is provided and directly implies that \kcol is also hard on $\mathcal T_4$, for any fixed $k \geqslant 3$.

On the algorithmic side, most graph problems are in \P\ on~$\mathcal T_1$ since this class has bounded clique-width~\cite{twin-width6} and is a~subclass of permutation graphs~\cite{AhnTwinWidthOne}.
Although 1-sequences can be computed in polynomial time in~$\mathcal T_1$~\cite{tww-polyker}, and even in linear time~\cite{AhnTwinWidthOne}, the corresponding algorithms do not require a~sequence to be provided as part of the input. 
\mis (and consequently, \vc and \cli) can be solved in polynomial time in~$\mathcal T_2$ if 2-sequences are provided as part of the input~\cite{twin-width3}.

While $\mathcal T_1$ can be recognized in polynomial time, it is \NP-hard to decide if a~graph (of~$\mathcal T_5$) is in $\mathcal T_4$.
The complexity of the recognition of~$\mathcal T_2$ and of~$\mathcal T_3$ is open.

\begin{table}[!ht]
  \centering
  \begin{tabular}{lcccc}
    \toprule
    & $\mathcal T_1$ & $\mathcal T_2$ & $\mathcal T_3$ & $\mathcal T_4$ \\
    \midrule
    \chr & \P~\cite{AhnTwinWidthOne,EvenPL72} & ? & \NP-c~(\cref{thm:col-tww-3}) & \\
    \tcol & \P~\cite{twin-width6,CMR00} & ? & ? & \NP-c~(\cref{thm:tcol-hardness}) \\
    \kcol, $k \geqslant 3$ & \P~\cite{twin-width6,CMR00} & ? & ? & \NP-c~(\cref{cor:tcol-hardness}) \\
    \midrule
    \mis & \P~\cite{twin-width6,CMR00} & \textcolor{gray}{\P}~\cite{twin-width3} & ? & \NP-c~\cite{BergeBonnetDepres} \\
    \vc & \P~\cite{twin-width6,CMR00} & \textcolor{gray}{\P}~\cite{twin-width3} & ? & \NP-c~\cite{BergeBonnetDepres} \\
    \cli & \P~\cite{twin-width6,CMR00} & \textcolor{gray}{\P}~\cite{twin-width3} & ? & \NP-c~\cite{BergeBonnetDepres} \\
    \fvs & \P~\cite{twin-width6,CMR00} & ? & ? & \NP-c~\cite{BergeBonnetDepres} \\
    \mds & \P~\cite{twin-width6,CMR00} & ? & ? & \NP-c~\cite{BergeBonnetDepres} \\
    \mim & \P~\cite{twin-width6,CMR00} & ? & ? & \NP-c~\cite{BergeBonnetDepres} \\
    \hamp & \P~\cite{AhnTwinWidthOne,Damaschke91} & ? & ? & \NP-c~\cite{twin-width1,IPS82} \\
    \hamc & \P~\cite{AhnTwinWidthOne,DeogunS94} & ? & ? & \NP-c~\cite{twin-width1,IPS82} \\
    \textsc{Recognition} & \P~\cite{tww-polyker,AhnTwinWidthOne} & ? & ? & \NP-c~\cite{BergeBonnetDepres} \\
    \bottomrule
  \end{tabular}
   \caption{Complexity of some of the main \NP-complete graph problems in $\mathcal T_d$ for $d=1, 2, 3, 4$.
    The \textcolor{gray}{\P} in gray in the $\mathcal T_2$ column means that 2-sequences are required by the known algorithm.}
  \label{tab:summary}
\end{table}

\Cref{tab:summary} suggests the task of closing these gaps.

\begin{question}\label{q:gaps}
 Establish, for a~problem $\Pi$ of~\cref{tab:summary}, a~nonnegative integer~$d$ such that $\Pi$~on~$\mathcal T_d$ is in \P\ and $\Pi$~on~$\mathcal T_{d+1}$ is \NP-hard. 
\end{question}

To our knowledge, the only problems for which \cref{q:gaps} is settled are \textsc{Firefighter} and \textsc{Restricted Vertex Multicut}, which are in \P\ on~$\mathcal T_1$ (as they are polynomial-time solvable on permutation graphs~\cite{FominHL16,Papadopoulos12}) and \NP-hard on~$\mathcal T_2$ (as they are hard on trees~\cite{FinbowKMR07,Calinescu03}).

As the current paper is on graph coloring, we ask the following question.

\begin{question}\label{q:chr}
  Is \chr on~$\mathcal T_2$ in \P?
\end{question}

As we mentioned, \chr is the first problem to be shown \NP-hard on~$\mathcal T_3$ while being tractable on classes known to have twin-width at~most~3: cographs (which coincide with~$\mathcal T_0$~\cite{twin-width1}), trees (which are in~$\mathcal T_2$~\cite{twin-width1}), and unit circular-arc graphs (which are in~$\mathcal T_3$~\cite[below Theorem 5.5]{Bonnet24hdr}).
Without this requirement, other examples would include \textsc{Achromatic Number}, which is \NP-hard on cographs \cite{Bodlaender89}, \bandw, which is \NP-hard on trees~\cite{Garey78}, and \textsc{Partial Representation Extension}, which is \NP-hard on unit circular-arc graphs~\cite{Fiala22}.
\Cref{thm:col-tww-3} thus involved a~novel encoding that incurred a~quadratic blow-up. 
We wonder whether this blow-up can be avoided. 

\begin{question}\label{q:subexp}
  Can \chr be solved in time $2^{O(\sqrt{n})}$ on $n$-vertex graphs of~$\mathcal T_3$?
\end{question}

The class~$\mathcal T_1$ is now well understood~\cite{AhnTwinWidthOne}.
In particular, it is a~subclass of permutation graphs with bounded clique-width.
We believe that the classes $\mathcal T_2$ and $\mathcal T_3$ have some still hidden structure.
While the membership problem in~$\mathcal T_4$ is~\NP-hard~\cite{BergeBonnetDepres}, that of $\mathcal T_2$ and of~$\mathcal T_3$ are open.
It was proven that \emph{weakly sparse} subclasses (i.e., excluding a~biclique $K_{t,t}$ as a~subgraph) of $\mathcal T_2$ have bounded treewidth, whereas this is not the case for~$\mathcal T_3$~\cite{BergougnouxG23}.

However, bounded-degree subclasses of~$\mathcal T_3$ have bounded treewidth. 
This is because bounded-degree graphs of large treewidth admit subdivisions of large walls or their line graphs as induced subgraphs~\cite{Korhonen23} and those graphs have twin-width~(exactly)~4~\cite{Ahn25}.
Furthermore, the latter result combined with \cite{Abrishami24,Bonnet25} (extending~\cite{Korhonen23}) implies that every subclass of~$\mathcal T_3$ without large subdivided cliques as subgraphs has bounded treewidth.
As~\tcol is in \P\ on any class of bounded treewidth, if \tcol is \NP-hard on~$\mathcal T_3$, then the hard instances must contain arbitrarily large clique subdivisions as subgraphs. 
The hard \tcol instances that the proof of~\cref{thm:tcol-hardness} builds do not: they have a~single vertex of degree more than 4.

\begin{question}\label{q:3-3}
  Is \tcol on~$\mathcal T_3$ in \P?
\end{question}

The same question holds for~\mis.

\begin{question}\label{q:mis-3}
  Is \mis on~$\mathcal T_3$ in \P?
\end{question}

Another consequence of the above result of~\cite{Ahn25} is that $\mathcal T_3$ excludes a~fixed planar graph as an induced minor.
Thus \cref{q:mis-3,q:3-3} are special cases of the same questions on any class excluding a~planar induced minor, which was previously raised for \mis~\cite{Dallard24}.
The class $\mathcal T_3$ is a~good candidate for a~negative answer to the latter question (on the pessimistic side), or \cref{q:mis-3} could serve as a~preliminary step in positively answering it (on the optimistic side).
There are other questions and conjectures on classes excluding a~planar induced minor.
They can be revisited on the particular class~$\mathcal T_3$, like for instance the following conjecture appearing in~\cite{GartlandThesis}.

\begin{question}\label{q:bal-sep}
 Is there a~universal constant~$c$ such that every graph of~$\mathcal T_3$ admits a~balanced separator included in the neighborhood of at~most~$c$ vertices?
\end{question}

We recall the relevant definitions and notation in~\cref{sec:prelim}, show~\cref{thm:col-tww-3} in \cref{sec:col-tww-3}, and \cref{thm:tcol-hardness} in \cref{sec:tcol-hardness}.

\section{Preliminaries}\label{sec:prelim}

For two integers $i$ and $j$, we denote by $[i,j]$ the set of integers that are at least $i$ and at most~$j$.
For every integer $i$, $[i]$ is a shorthand for $[1,i]$.

\subsection{Standard graph-theoretic definitions and notation}

We denote by $V(G)$ and $E(G)$ the set of vertices and edges of a graph $G$, respectively.
For $S \subseteq V(G)$, the \emph{subgraph of $G$ induced by $S$}, denoted $G[S]$, is obtained by removing from $G$ all the vertices that are not in $S$.
Then $G-S$ is a shorthand for $G[V(G)\setminus S]$.
We denote by $N_G(v)$ the set of neighbors of~$v$ in~$G$.
A~\emph{subdivision} of a~graph $G$ is any graph $H$ obtained from $G$ by replacing edges $e$ of~$G$ by paths with at~least~one edge whose extremities are the endpoints of~$e$.
A~\emph{$(\geqslant s)$-subdivision} is a~subdivision where every edge is replaced by a~path with at~least~$s$ internal vertices. 

In this paper, a~coloring of a~graph is implicitly assumed to be a~proper vertex-coloring.

\subsection{Trigraphs, partition sequences, and twin-width}

A~\emph{trigraph} $G$ has vertex set $V(G)$, black edge set $E(G)$, red edge set $R(G)$ such that $E(G) \cap R(G) = \emptyset$ (and $E(G),R(G) \subseteq {V(G) \choose 2}$).
Two vertices $u, v$ such that $uv \in R(G)$ are called \emph{red neighbors}.
The \emph{red degree} of~$u$ is its number of red neighbors.
The \emph{maximum red degree} of~$G$ is the maximum red degree among all its vertices.

Given a~(tri)graph $G$ and a partition $\mathcal P$ of $V(G)$, the \emph{quotient trigraph} $G/\mathcal P$ is the trigraph with vertex set $\mathcal P$, where $PP'$ is a~black edge if these two parts are fully adjacent via black edges (i.e., for every $u \in P$ and every $v \in P'$, $uv \in E(G)$), and a~red edge if there is $u \in P$ and $v \in P'$ such that $uv \in R(G)$ or $u_1, u_2 \in P$ and $v_1, v_2 \in P'$ such that $u_1v_1 \in E(G)$ and $u_2v_2 \notin E(G)$.

A~\emph{partition sequence} of an $n$-vertex (tri)graph~$G$ is a~sequence $\mathcal P_n, \mathcal P_{n-1}, \ldots, \mathcal P_1$ of partitions of $V(G)$ such that $\mathcal P_n = \{\{v\}:v \in V(G)\}$, $\mathcal P_1 = \{V(G)\}$, and for every $i \in [n-1]$, $\mathcal P_i$ is obtained from $\mathcal P_{i+1}$ by merging $P, P' \in \mathcal P_{i+1}$ into $P \cup P'$.
A~\emph{$d$-sequence} of~$G$ is a~partition sequence $\mathcal P_n, \ldots, \mathcal P_1$ such that for every $i \in [n]$, the maximum red degree of~$G/\mathcal P_i$ is at~most~$d$.
The twin-width of a~(tri)graph is the least integer~$d$ such that it admits a~$d$-sequence.

A~partition sequence is called \emph{partial} when we relax the condition that the last partition has a~single part.
We say that a~trigraph is \emph{fully red} if it does not have any black edges.
We will use the simple fact that turning some (or all) black edges red cannot decrease the twin-width of a~trigraph.
Therefore, when showing twin-width upper bounds, we may sometimes assume that all the edges (black and red) are in fact red, if this simplifies a~later argument.

\section{Hardness of Coloring Graphs of Twin-Width 3}\label{sec:col-tww-3}

As a~decision problem, \chr takes as input a~graph $G$ and an integer $k$ and asks whether the chromatic number of $G$, $\chi(G)$, is at~most~$k$.
As a~function problem, one is given the mere graph $G$ and has to output a~(proper) $\chi(G)$-coloring of~$G$.
We show that \chr is \NP-hard on~$\mathcal T_3$, even in its decision form and when a~3-sequence of~$G$ is given as input.
As \chr is polynomial-time solvable on cographs, on trees, and on unit circular-arc graphs (all of which have twin-width at~most~3), and planar graphs are not included in~$\mathcal T_3$, this requires finding a~novel kind of encoding of rich structures (here, \textsc{3-Sat} instances) onto~graphs of twin-width at~most~3.   

\begin{reptheorem}{thm:col-tww-3}
  \chr is \NP-hard and, unless the ETH fails, requires $2^{\Omega(\sqrt{N})}$ time on $N$-vertex graphs of twin-width at~most~3 (even if a~3-sequence is provided). 
\end{reptheorem}

\begin{proof}
  We reduce from \textsc{3-Sat}, which is \NP-complete~\cite{Karp72}, and unless the ETH fails, its $n$-variable $m$-clause instances cannot be solved in time $2^{o(n+m)}$ by the so-called Sparsification Lemma~\cite{sparsification}; it is indeed shown that, under the ETH, $n$-variable \textsc{3-Sat} cannot be solved in time $2^{o(n)}$ even on instances with $O(n)$ clauses. 
  Let $x_1, \ldots, x_n$ be the variables of a~\textsc{3-Sat} instance~$\varphi$, and $c_1, \ldots, c_m$ be its clauses.

  We describe the construction of a~graph $G := G(\varphi)$ with chromatic number (at most) $2n$ if and only if $\varphi$ is satisfiable.
  We start with two paths $P, P'$ each on $2np$ vertices with $p := 2n+m$, which we partition evenly into $p$ subpaths on $2n$ vertices.
  Let $A_1, A_2, \ldots, A_p$ (resp.~$B_1, B_2, \ldots, B_p$) be the vertex sets of the subpaths of~$P$ (resp.~$P'$) from left to right.
  We set $A := \bigcup_{i \in [p]} A_i = V(P)$ and $B := \bigcup_{i \in [p]} B_i = V(P')$.
  We denote by $a_{i,1}, a_{i,2}, \ldots, a_{i,2n}$ the $2n$ vertices of the subpath $P[A_i]$ from left to right.
  Similarly we denote by $b_{i,1}, b_{i,2}, \ldots, b_{i,2n}$ the $2n$ vertices of the subpath $P'[B_i]$ from left to right.
  We now define $G[A]$ (resp.~$G[B]$) as $P^{\leq 2n-1}$ (resp.~$P'^{\leq 2n-1}$), that is, for every $u,v \in A$ (resp. $u,v \in B$), $uv \in E(G)$ whenever $u$ and $v$ are at distance at~most~$2n-1$ in~$P$ (resp.~in~$P'$).
  In particular, each $A_i$ (and each $B_i$) is a~clique of size $2n$, and for all $i,i' \in [p]$ with $i < i'$, it holds that $a_{i,j}$ and $a_{i',j'}$ are adjacent if and only if $i' = i+1$ and $j' < j$.    
  There is no edge between $A$ and $B$.

  For every $i \in [p]$, we add one vertex $v_i$ adjacent to a~subset of~$A_i \cup B_i$, as follows.
  \begin{itemize}
  \item For every $i \in [n]$, $N_G(v_{2i-1}) = (A_{2i-1} \cup B_{2i-1}) \setminus \{a_{2i-1,2i-1}, a_{2i-1,2i}, b_{2i-1,2i-1}\}$ and $N_G(v_{2i})=(A_{2i} \cup B_{2i}) \setminus \{a_{2i,2i-1}, a_{2i,2i}, b_{2i,2i}\}$.
  \item For every $i \in [2n+1,2n+m]$, say $c_{i-2n}$ is the clause $s_1 x_{j_1} \lor s_2 x_{j_2} \lor s_3 x_{j_3}$ (repeat a~literal if it had only two), where $s_1, s_2, s_3$ are signs in $\{\neg,\varepsilon\}$ ($\varepsilon$ is the positive sign).
   Then, $N_G(v_i)=(B_i \setminus \{b_{i, 2j_1}, b_{i, 2j_2}, b_{i, 2j_3}\}) \cup \{a_{i, 2j_1-f(s_1)}, a_{i, 2j_2-f(s_2)}, a_{i, 2j_3-f(s_3)}\}$ with $f(\neg)=0$ and $f(\varepsilon)=1$.
  \end{itemize}
  \medskip
  This finishes the construction of~$G$; see~\cref{fig:col-hardness} for an illustration.
  Note that $G$ has $N:=(4n+1)p=O(n(n+m))$ vertices, and that it can be constructed in polynomial time from~$\varphi$.
  Thus, provided the reduction is correct, which we next prove, a~$2^{o(\sqrt{N})}$-time algorithm for \chr on the produced instances would imply a~$2^{o(n+m)}$-time algorithm for $n$-variable $m$-clause \textsc{3-Sat}, hence refute the ETH.
  
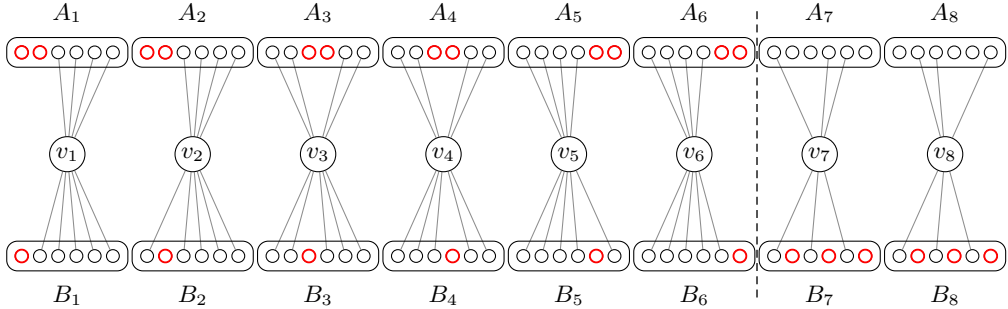
\begin{figure}[!ht]
\centering
\begin{tikzpicture}[
  font=\small,
  >=Stealth,
  vertex/.style={circle,draw,fill=white,inner sep=0.2pt,minimum size=4.8pt},
  missing/.style={circle,draw=red,line width=0.7pt,fill=white,inner sep=0.2pt,minimum size=4.8pt},
  block/.style={rounded corners,draw,inner sep=3pt},
  edge/.style={line width=0.4pt,opacity=0.45},
  vnode/.style={circle,draw,fill=white,inner sep=1.2pt}
]

\def\twon{6}          
\def\pblocks{8}       

\def\dx{1.65cm}       
\def\vx{0.24cm}       
\def\yA{1.35cm}
\def\yB{-1.35cm}

\foreach \i in {1,...,\pblocks}{
  \foreach \j in {1,...,\twon}{
    \node[vertex] (a-\i-\j)
      at ({(\i-1)*\dx + (\j-0.5-\twon/2)*\vx},{\yA}) {};
  }
  \node[block,fit=(a-\i-1) (a-\i-\twon)] (A-\i) {};
  \node[above=2pt of A-\i] {$A_{\i}$};

  \foreach \j in {1,...,\twon}{
    \node[vertex] (b-\i-\j)
      at ({(\i-1)*\dx + (\j-0.5-\twon/2)*\vx},{\yB}) {};
  }
  \node[block,fit=(b-\i-1) (b-\i-\twon)] (B-\i) {};
  \node[below=2pt of B-\i] {$B_{\i}$};

  \node[vnode] (v-\i) at ({(\i-1)*\dx},{0}) {$v_{\i}$};
}

\draw[densely dashed,line width=0.5pt]
  ({5.5*\dx},{\yB-0.55cm}) -- ({5.5*\dx},{\yA+0.55cm});

\foreach \i in {1,...,\twon}{
  \foreach \j in {1,...,\twon}{
    \ifnum\j=\i\relax
    \else
      \draw[edge] (v-\i) -- (b-\i-\j);
    \fi

    \ifodd\i
      \ifnum\j=\i\relax
      \else\ifnum\j=\numexpr\i+1\relax\relax
      \else
        \draw[edge] (v-\i) -- (a-\i-\j);
      \fi\fi
    \else
      \ifnum\j=\numexpr\i-1\relax\relax
      \else\ifnum\j=\i\relax
      \else
        \draw[edge] (v-\i) -- (a-\i-\j);
      \fi\fi
    \fi
  }

  \node[missing] at (b-\i-\i) {};
  \ifodd\i
    \node[missing] at (a-\i-\i) {};
    \node[missing] at (a-\i-\the\numexpr\i+1\relax) {};
  \else
    \node[missing] at (a-\i-\the\numexpr\i-1\relax) {};
    \node[missing] at (a-\i-\i) {};
  \fi
}


\def\ci{7}
\foreach \j in {1,...,\twon}{
  \ifnum\j=2\relax\else
  \ifnum\j=4\relax\else
  \ifnum\j=6\relax\else
    \draw[edge] (v-\ci) -- (b-\ci-\j);
  \fi\fi\fi
}
\node[missing] at (b-\ci-2) {};
\node[missing] at (b-\ci-4) {};
\node[missing] at (b-\ci-6) {};
\draw[edge] (v-\ci) -- (a-\ci-1);
\draw[edge] (v-\ci) -- (a-\ci-4);
\draw[edge] (v-\ci) -- (a-\ci-5);

\def\ci{8}
\foreach \j in {1,...,\twon}{
  \ifnum\j=2\relax\else
  \ifnum\j=4\relax\else
  \ifnum\j=6\relax\else
    \draw[edge] (v-\ci) -- (b-\ci-\j);
  \fi\fi\fi
}
\node[missing] at (b-\ci-2) {};
\node[missing] at (b-\ci-4) {};
\node[missing] at (b-\ci-6) {};
\draw[edge] (v-\ci) -- (a-\ci-2);
\draw[edge] (v-\ci) -- (a-\ci-3);
\draw[edge] (v-\ci) -- (a-\ci-6);

\end{tikzpicture}
\caption{The graph $G(\varphi)$ for $\varphi$ consisting of two clauses $c_1 = x_1 \lor \neg x_2 \lor x_3$ and $c_2 = \neg x_1 \lor x_2 \lor \neg x_3$.
  The variable gadgets and clause gadgets are delimited by the vertical dashed line.
  Not to clutter the picture, we have not drawn the edges in $G[A \cup B]$.}
\label{fig:col-hardness}
\end{figure}

  \medskip

  \textbf{If $\bm{\varphi}$ is satisfiable, then $\bm{G}$ is $\bm{2n}$-colorable.}
  Let $\mathcal A$ be a~satisfying assignment for~$\varphi$.
  We then describe a~(proper) $2n$-coloring $c$ of~$G$.
  \begin{itemize}
    \item For every $i \in [p]$ and every $j \in [2n]$, we set $c(b_{i,j})=j$.
    \item For every $i \in [p]$ and every $j \in [n]$, we set $c(a_{i,2j-1})=2j-1$ and $c(a_{i,2j})=2j$ if $\mathcal A$~assigns $x_j$ to true, or $c(a_{i,2j})=2j-1$ and $c(a_{i,2j-1})=2j$, otherwise.
  \end{itemize}
  \medskip
  As the $2n$ colors repeat with period $2n$ along each path of $\{P,P'\}$, no edge of $G[A \cup B]$ is monochromatic.
  We now simply need to argue that for every vertex $v_i$, there is at least one color $c'$ of $[2n]$ not present in $N_G(v_i)$; and set $c(v_i)=c'$.

  For every $i \in [2n]$, we can set $c(v_i)=i$, as this color does not appear in~$N_G(v_i)$.
  For every $i \in [2n+1,2n+m]$, with $c_{i-2n} = s_1 x_{j_1} \lor s_2 x_{j_2} \lor s_3 x_{j_3}$ and $s_1, s_2, s_3 \in \{\neg,\varepsilon\}$, there is an $h \in \{1,2,3\}$ such that $\mathcal A$ satisfies $s_h x_{j_h}$.
  Then we can set $c(v_i)=2j_h$.
  Indeed, if $\mathcal A$ sets $x_{j_h}$ to true, then $f(s_h)=1$ and $a_{i, 2j_h-1}$ is colored $2j_h-1$, whereas if $\mathcal A$ sets $x_{j_h}$ to false, then $f(s_h)=0$ and $a_{i, 2j_h}$ is colored $2j_h-1$; thus, in both cases, color $2j_h$ is still available for~$v_i$.
  See \cref{fig:col-hardness-colored} for an example.

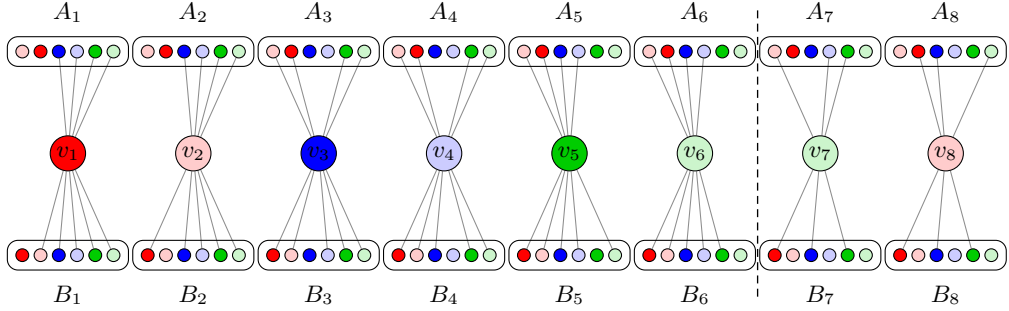
\begin{figure}[!ht]
\centering
\begin{tikzpicture}[
  font=\small,
  >=Stealth,
  vertex/.style={circle,draw,fill=white,inner sep=0.2pt,minimum size=4.8pt},
  missing/.style={circle,draw=red,line width=0.7pt,fill=none,inner sep=0.2pt,minimum size=4.8pt},
  block/.style={rounded corners,draw,inner sep=3pt},
  edge/.style={line width=0.4pt,opacity=0.45},
  vnode/.style={circle,draw,fill=white,inner sep=1.2pt}
]


\def\twon{6}          
\def\pblocks{8}       

\def\dx{1.65cm}       
\def\vx{0.24cm}       
\def\yA{1.35cm}
\def\yB{-1.35cm}

\colorlet{col1}{red}
\colorlet{col2}{red!20}
\colorlet{col3}{blue}
\colorlet{col4}{blue!20}
\colorlet{col5}{green!80!black}
\colorlet{col6}{col5!20}

\foreach \i in {1,...,\pblocks}{
  \foreach \j in {1,...,\twon}{
    \pgfmathtruncatemacro{\acol}{\j}
    \ifnum\j=1\relax\pgfmathtruncatemacro{\acol}{2}\fi
    \ifnum\j=2\relax\pgfmathtruncatemacro{\acol}{1}\fi
    \edef\Afill{col\acol}
    \node[vertex,fill=\Afill] (a-\i-\j)
      at ({(\i-1)*\dx + (\j-0.5-\twon/2)*\vx},{\yA}) {};
  }
  \node[block,fit=(a-\i-1) (a-\i-\twon)] (A-\i) {};
  \node[above=2pt of A-\i] {$A_{\i}$};

  \foreach \j in {1,...,\twon}{
    \pgfmathtruncatemacro{\bcol}{\j}
    \edef\Bfill{col\bcol}
    \node[vertex,fill=\Bfill] (b-\i-\j)
      at ({(\i-1)*\dx + (\j-0.5-\twon/2)*\vx},{\yB}) {};
  }
  \node[block,fit=(b-\i-1) (b-\i-\twon)] (B-\i) {};
  \node[below=2pt of B-\i] {$B_{\i}$};

  \pgfmathtruncatemacro{\vcol}{\i}
  \ifnum\i=7\relax\pgfmathtruncatemacro{\vcol}{6}\fi 
  \ifnum\i=8\relax\pgfmathtruncatemacro{\vcol}{2}\fi 
  \edef\Vfill{col\vcol}
  \node[vnode,fill=\Vfill] (v-\i) at ({(\i-1)*\dx},{0}) {$v_{\i}$};
}

\draw[densely dashed,line width=0.5pt]
  ({5.5*\dx},{\yB-0.55cm}) -- ({5.5*\dx},{\yA+0.55cm});

\foreach \i in {1,...,\twon}{
  \foreach \j in {1,...,\twon}{
    \ifnum\j=\i\relax
    \else
      \draw[edge] (v-\i) -- (b-\i-\j);
    \fi

    \ifodd\i
      \ifnum\j=\i\relax
      \else\ifnum\j=\numexpr\i+1\relax\relax
      \else
        \draw[edge] (v-\i) -- (a-\i-\j);
      \fi\fi
    \else
      \ifnum\j=\numexpr\i-1\relax\relax
      \else\ifnum\j=\i\relax
      \else
        \draw[edge] (v-\i) -- (a-\i-\j);
      \fi\fi
    \fi
  }
}

\def\ci{7}
\foreach \j in {1,...,\twon}{
  \ifnum\j=2\relax\else
  \ifnum\j=4\relax\else
  \ifnum\j=6\relax\else
    \draw[edge] (v-\ci) -- (b-\ci-\j);
  \fi\fi\fi
}
\node at (b-\ci-2) {};
\node at (b-\ci-4) {};
\node at (b-\ci-6) {};
\draw[edge] (v-\ci) -- (a-\ci-1);
\draw[edge] (v-\ci) -- (a-\ci-4);
\draw[edge] (v-\ci) -- (a-\ci-5);

\def\ci{8}
\foreach \j in {1,...,\twon}{
  \ifnum\j=2\relax\else
  \ifnum\j=4\relax\else
  \ifnum\j=6\relax\else
    \draw[edge] (v-\ci) -- (b-\ci-\j);
  \fi\fi\fi
}
\node at (b-\ci-2) {};
\node at (b-\ci-4) {};
\node at (b-\ci-6) {};
\draw[edge] (v-\ci) -- (a-\ci-2);
\draw[edge] (v-\ci) -- (a-\ci-3);
\draw[edge] (v-\ci) -- (a-\ci-6);

\end{tikzpicture}
\caption{A~(proper) $2n$-coloring of the graph $G(\varphi)$ of~\cref{fig:col-hardness} corresponding to the satisfying assignment $\mathcal A = \{x_1 \mapsto \text{false}, x_2 \mapsto \text{true}, x_3 \mapsto \text{true}\}$.
The coloring of $v_7$ (resp.~$v_8$) witnesses that $c_1$ (resp.~$c_2$) is satisfied by literal $x_3$ (resp.~$\neg x_1$).}
\label{fig:col-hardness-colored}
\end{figure}

  \medskip

  \textbf{If $\bm{G}$ is $\bm{2n}$-colorable, then $\bm{\varphi}$ is satisfiable.}
  Fix a~(proper) coloring $c \colon V(G) \to [2n]$ of~$G$, with $c(S) := \{c(v):v \in S\}$ for any $S \subseteq V(G)$.
  As $B_1$ is a~clique, we can further assume that $c(b_{1,j})=j$ for every $j \in [2n]$. 
  By a~straightforward induction, for every $i \in [p-1]$ and $j \in [2n]$, $c(a_{i,j})=c(a_{i+1,j})$ and $c(b_{i,j})=c(b_{i+1,j})=j$, and for every $i \in [p]$, $c(A_i)=[2n]=c(B_i)$.
  Indeed both $A_1$ and $B_1$ are $2n$-vertex cliques, thus $c(A_1)=c(B_1)=[2n]$, and when $2n$-coloring $P^{\leqslant 2n-1}$ (resp.~$P'^{\leqslant 2n-1}$) from left to right, the only available color for $a_{i+1,j}$ (resp.~$b_{i+1,j}$) is $c(a_{i,j})$ (resp.~$c(b_{i,j})$).

  We first show that, for every $i \in [p]$ and $j \in [n]$, $\{c(a_{i,2j-1}),c(a_{i,2j})\}=\{2j-1,2j\}$.
  Assume for contradiction that there is an $h \in \{c(a_{i,2j-1}),c(a_{i,2j})\} \setminus \{2j-1,2j\}$.
  By the previous paragraph, observe that $h$ does not depend on~$i$.
  This implies that $2j-1 \notin \{c(a_{i,2j-1}),c(a_{i,2j})\}$ or $2j \notin \{c(a_{i,2j-1}),c(a_{i,2j})\}$.
  In the former case, $v_{2j-1}$ would need a~$(2n+1)$-st color, whereas in the latter case, the same would happen to~$v_{2j}$.

  Let $\mathcal A$ be the truth assignment that sets, for every $j \in [n]$, $x_j$ to true if $c(a_{i,2j-1})=2j-1$ and $c(a_{i,2j})=2j$, and to false, if $c(a_{i,2j-1})=2j$ and $c(a_{i,2j})=2j-1$ (again, note that this does not depend on~$i$).
  Let us show that $\mathcal A$ satisfies all the clauses of~$\varphi$.
  For each $i \in [2n+1,2n+m]$, say $c_{i-2n} = s_1 x_{j_1} \lor s_2 x_{j_2} \lor s_3 x_{j_3}$ with $s_1, s_2, s_3 \in \{\neg,\varepsilon\}$.
  Since $v_i$ is adjacent to $B_i \setminus \{b_{i,2j_1}, b_{i,2j_2}, b_{i,2j_3}\}$, it holds that $c(v_i) \in \{2j_1, 2j_2, 2j_3\}$; say, $c(v_i) = 2j_h$ with $h \in \{1,2,3\}$.
  This implies that $c(a_{i,2j_h-f(s_h)}) \neq 2j_h$, thus $c(a_{i,2j_h-f(s_h)}) = 2j_h-1$.
  Therefore, $s_h=\varepsilon$ and $\mathcal A$ sets $x_{j_h}$ to true, or $s_h=\neg$ and $\mathcal A$ sets $x_{j_h}$ to false.
  In both cases, the literal $s_h x_{j_h}$ of $c_{i-2n}$ is satisfied by~$\mathcal A$.

  \medskip

  \textbf{$\bm{G}$ admits a~3-sequence.}
  We describe a~partition sequence for $G$ of width at~most~3.
  The sequence has two stages.
  In the first stage, each $A_i$ and each $B_i$ is merged into a~single part.
  Throughout this stage, we denote by $P_i$ (resp.~$P'_i$) the current part containing $a_{i,1}$ (resp.~$b_{i,1}$).
  For $j$ going from 2 to $2n$, for $i$ going from 1 to $p$, merge $P_i$ and singleton $\{a_{i,j}\}$, and merge $P'_i$ and singleton $\{b_{i,j}\}$.

  Let us argue that this partial sequence has width at~most~3.
  At any point, the singleton parts $\{a_{i,j+1}\}, \ldots, \{a_{i,2n}\}$ \emph{in between} $P_i = \{a_{i,1}, \ldots, a_{i,j}\}$ and $P_{i+1}$ are all fully adjacent to $P_i \cup P_{i+1}$, some of these singletons are adjacent to $v_i$, and they have no other adjacencies.
  Thus these parts have red degree 0.
  Note that the singleton parts $\{a_{p,j+1}\}, \ldots, \{a_{p,2n}\}$ do not have incident red edges either.
  Each singleton part $\{v_i\}$ has red degree at~most~2, without any red neighbor outside $\{P_i, P'_i\}$.
  Finally, at any point, each $P_i$ has at most three red neighbors: $P_{i-1}$ (if it exists), $v_i$, and $P_{i+1}$ (if it exists).
  Indeed, all the other parts are either fully nonadjacent to $P_i$ or have red degree~0.
  Symmetrically, $P'_i$ has always at most three red neighbors.

  After the first stage is completed, the resulting trigraph is fully red (we may well assume that $n \geqslant 2$), has $3p$ vertices, and consists of two $p$-vertex paths whose $i$th vertices have an additional shared neighbor of degree~2, for every $i \in [p]$.  
  In the second stage, we finish the 3-sequence, as follows.
  Let us call $Q$ the part containing $v_1$. 
  For $i$ going from 2 to $p$, we merge $P_1$ and $P_i$, and keep denoting the resulting part by $P_1$, we merge $P'_1$ and $P'_i$, and denote the resulting part by $P'_1$, and then we merge $Q$ and $\{v_i\}$.
  Finally, only three parts remain that we merge to a~single part in any way.
  One can observe that at any time of the second stage when at~least~three parts are left, every part $P_i$ or $P'_i$ (including $P_1, P'_1$) has at~most~three red neighbors, and $Q$, as well as every part $\{v_i\}$, has two red neighbors. 
  So this finishes the description of a~3-sequence of~$G(\varphi)$.

  As this sequence can be computed in polynomial time for any formula~$\varphi$, the hardness of \chr on~$\mathcal T_3$ further holds when a~3-sequence is provided as part of the input.
\end{proof}

\section{Hardness of 3-Coloring Graphs of Twin-Width 4}\label{sec:tcol-hardness}

The hardness proof of the previous section requires an unbounded number of colors.
In this section we show that deciding if a~graph of twin-width at~most~4 is 3-colorable is \NP-hard.
We also get a~similar ETH lower bound, and readily extend this hardness result to \textsc{$k$-Coloring}.

\begin{reptheorem}{thm:tcol-hardness}
  \tcol is \NP-hard and, unless the ETH fails, requires $2^{\Omega(n/\log n)}$ time on $n$-vertex graphs of twin-width at~most~4 (even if a~4-sequence is provided).  
\end{reptheorem}

\begin{proof}
  It is known that \tcol is \NP-hard on $n$-vertex graphs of maximum degree at~most~4, and requires $2^{\Omega(n)}$ time, unless the ETH fails~(see for instance \cite[Lemma 1]{CyganFGKMPS17}).
  
  Let $G$ be an $n$-vertex graph of degree at~most~4.
  We build an $O(n \log n)$-vertex graph $G'$ of twin-width at~most~4 such that $G'$ is 3-colorable if and only if $G$ is 3-colorable.
  Let $q := 2 \lceil \log n \rceil$.
  Graph $G'$ is obtained from $G$ by subdividing each edge $e \in E(G)$ $q$ times, thereby creating vertices $v_{e,1}, \ldots, v_{e,q}$, and then adding one true twin (i.e., an adjacent vertex with the same closed neighborhood) $v'_{e,1}, v'_{e,3}, v'_{e,5}, \ldots, v'_{e,q-1}$ to $v_{e,1}, v_{e,3}, v_{e,5}, \ldots, v_{e,q-1}$, respectively (that is, every other created vertex gets a~true twin).
  The ``orientation'' of $e$ to fully specify this process is irrelevant, so we leave it undefined.
  (To get a~unique construction, one can totally order $V(G)$, and always ``orient'' the edge from its smaller endpoint to its larger endpoint.)

  The way $G'$ is built, it can be observed that there is a~partial 0-sequence from $G'$ to the $q$-subdivision of~$G$.
  Since every $(\geqslant 2 \log n)$-subdivision of an $n$-vertex graph has twin-width at~most~4~\cite{BergeBonnetDepres}, $G'$ has twin-width at~most~4.
  The number of vertices of $G'$ is $O(n+|E(G)| \log n)=O(n \log n)$, and its construction can be done in $O(n \log n)$ time.

  Given an original vertex $u$ (i.e., $u \in V(G)$) of~$G'$ adjacent to some $v_{e,1}$, and a~color $i \in [3]$, the \emph{3-coloring propagation} of $u$ along~$e$ is the only (proper) 3-coloring of \[G[\{u,v_{e,1},v'_{e,1},v_{e,2},v_{e,3},v'_{e,3}, \ldots, v_{e,q-1},v'_{e,q-1},v_{e,q}\}]\] that gives color $i$ to~$u$.
  Note that this 3-coloring also gives color $i$ to $v_{e,q}$.
  
  If $G$ is 3-colorable, we start by copying a~3-coloring of~$G$ at the original vertices of~$G'$.
  We then make the 3-coloring propagation for each pair $u \in V(G), v_{e,1}$ with $uv_{e,1} \in E(G')$.
  By the previous remark, this defines a~(proper) 3-coloring of~$G'$.
  Conversely, the projection of any 3-coloring of~$G'$ onto $V(G)$ properly colors~$G$.
\end{proof}

\begin{corollary}\label{cor:tcol-hardness}
  For every $k \geqslant 3$, \kcol is \NP-hard and, unless the ETH fails, requires $2^{\Omega(n/\log n)}$ time on $n$-vertex graphs of twin-width at~most~4.  
\end{corollary}
\begin{proof}
  Adding $k-3$ universal vertices to the graph $G'$ built in the proof of~\cref{thm:tcol-hardness} does not change its twin-width and makes the resulting graph $k$-colorable if and only if $G'$ is 3-colorable.
\end{proof}

\end{document}